\documentclass[preprint,12pt]{elsarticle}

\usepackage{graphicx,multicol,amsmath,amssymb}
\usepackage{mathtools}
\usepackage{tipa}
\graphicspath{{./experimente/}}
\DeclareGraphicsExtensions{.png}
\newcommand*\samethanksMihaiSEpsilon[2][\value{footnote}]{\footnotemark[#1]}

\usepackage[nodots,nocompress]{numcompress}

\begin{document}

\begin{frontmatter}



\title{Characterization and Detection of  $\epsilon$-Berge Zhukovskii Equilibria}

%\author{No\'{e}mi Gask\'{o}}
%\author[dave]{Mihai Suciu},
%\author[paul]{Rodica Ioana Lung}, and
%\author[gwen]{D. Dumitrescu}

%\address{Babe\c s-Bolyai University of Cluj-Napoca, Romania}
\author{No\'{e}mi Gask\'{o}\thanks{all authors have equal contribution}, Mihai Suciu\samethanksMihaiSEpsilon1, Rodica Ioana Lung\samethanksMihaiSEpsilon1, D. Dumitrescu\samethanksMihaiSEpsilon1}

\address{Babe\c s-Bolyai University of Cluj-Napoca, Romania \\ %Faculty of Mathematics and Computer Science\\
No. 1 Mihail Kogalniceanu Street,\\
RO-400084 Cluj-Napoca, Romania }

\begin{abstract}
Berge equilibrium in the sense of Zhukovskii (Berge-Zhu\-kov\-skii) is an alternate solution concept in non-cooperative game theory that formalizes cooperation in a noncooperative setting. In this paper the $\epsilon$-Berge-Zhukovskii equilibrium is introduced and characterized by using a generative relation. A computational  method for detecting $\epsilon$-Berge-Zhukovskii equilibrium based on evolutionary multiobjective optimization algorithms is presented. Numerical examples are used to illustrate the results obtained.

\end{abstract}


\end{frontmatter}

% \linenumbers

%==============================================================================
\section{Introduction}

\newtheorem{definition}{Definition}
\newtheorem{remark}{Remark}
\newtheorem{proposition}{Proposition}
\newtheorem{example}{Example}
\newtheorem{proof}{Proof}

In non-cooperative game theory players make decisions independently  based on their own interests and different equilibrium concepts are used to provide decision makers an overview over possible outcomes of the game. The most popular equilibrium concept is the Nash equilibrium \cite{Nash1951} - a game situation from which no player has an incentive to unilaterally deviate. A more general equilibrium concept based on the notion of equilibrium for a partition with respect to a coalition was proposed by Berge \cite{Berge1957}. Zhukovskii \cite{Zhukovskii1994} formalized a particularization of the general Berge equilibrium and introduced it as an alternate solution concept that is complementary to Nash and more suited for games where the Nash equilibrium has no practical value (trust games \cite{trustgames}, taxation games \cite{RePEc:lam:wpaper:11-05,Curt2}). In literature the equilibrium defined by Zhukovskii is referred to as Berge equilibrium in the sense of  Zhukovskii \cite{larbani2008,Nessah2007} or 
simply the Berge equilibrium \cite{RePEc:lam:wpaper:11-05,colman2011}. In this paper we will refer to it as the Berge-Zhukovskii equilibrium (BZ). 

The Berge-Zhukovskii (BZ) equilibrium is a situation in which every player's strategy is stable against the deviations of \textit{all} other players. Existence theorems and characterizations for BZ equilibrium can be found in  \cite{abalo2005,Nessah2007,Abalo20061840,RePEc:wsi:igtrxx:v:14:y:2012:i:01:p:1250005-1-1250005-10} and \cite{larbani2008}. A connection between the BZ and the Nash equilibrium of several two person games that provides also a method to find the BZ in $n$-player games is presented in \cite{colman2011}. To the best of our knowledge the first computational method aimed to directly detect the Berge-Zhukovskii equilibrium is described in \cite{gasko_berge}.

The $\epsilon$-Nash equilibrium  introduced by Radner \cite{radner} can be viewed  as a weakening of the strict rationality - in this case it is enough to be ``near''  to the Nash equilibrium - or to  approximate the Nash equilibrium. The value of $\epsilon$ can be interpreted in several ways: measuring an uncertainty of selecting a strategy, measuring a supplementary cost of attending the equilibrium strategy, or a perturbation of the players rationality \cite{epsilon_gecco2009}.

The $\epsilon$-Berge-Zhukovskii equilibrium is introduced in a similar manner. The intuition behind is the same as in the case of the $\epsilon$-Nash equilibrium: the $epsilon$ gives a perturbation to the players strategies. 

A generative relation for a certain game equilibrium is a binary relation defined on strategy profiles with the property that the set of strategy profiles non-dominated with respect to that relation is identic with the set of that equilibria of the game (non-dominated strategies are those for which there does not exists 'better' ones with respect to the generative relation).

A generative relation for describing $\epsilon$-Berge-Zhukovskii equilibrium is proposed. The most important feature of this relation is that it can be used to redirect the search of an evolutionary multiobjective optimization algorithm towards the $\epsilon$-Berge-Zhukovskii equilibria of the game. 
%The main idea is to use the proposed generative relation in a multi-objective evolutionary algorithm, in order to detect the $\epsilon$-Berge-Zhukovskii equilibrium. Numerical experiments are performed in order to illustrate the new equilibrium.

The remaining of the paper is structured as follows: after a brief introduction in non-cooperative games Section two, the Berge-Zhukovskii and $\epsilon$-Berge-Zhukovskii equilibrium are presented. Section three introduces the generative relation used in the equilibrium detection method. Section four gives a short introduction in multiobjective optimization. In Section five numerical experiments are described. The paper ends with conclusions.

%==============================================================================

\section{Berge-Zhukovskii and $\epsilon$-Berge-Zhukovskii equilibrium}

A finite strategic non-cooperative game can be formalized as a system $$G=((N,S_{i},u_{i}),i=1,...,n),$$ where:
\begin{itemize}
\item $N$ represents a set of players, and $n$ is the number of players;

\item $S_{i}$ is the set of actions available to player $i\in N$,   $$S=S_{1} \times S_{2} \times ... \times S_{n}$$
is the set of all possible situations of the game, and  $s=(s_1,...s_n) \in S$ is a strategy (or strategy profile) of the game;

\item for each player $i \in N$, $u_{i}:S \rightarrow R$ represents the payoff function of player $i$.

\end{itemize}

We will denote by  $S_{-i}=S_1 \times ... \times S_{i-1} \times S_{i+1} \times ... \times S_{n}$,  $s_{-i}=(s_1,...,s_{i-1},s_{i+1},...,s_n)$ and $( s^*_{i},s_{-i}^{})=(s_1,s_2,...,s_{i}^*,...,s_n).$

The Berge equilibrium is formally defined as follows:

\begin{definition}[Berge equilibrium]
Let $M$ be a finite set of indices. Denote by  $P=\{P_t\}, t\in M$ a partition of $N$ and $R=\{R_t\}, t \in M$ be a set of subsets of $N$. A strategy profile $s^{*} \in S$ is an equilibrium strategy for the partition $P$  with respect to the set $R$, or simply  a Berge equilibrium strategy, if and only if the condition
$$ u_{p_m}(s^*) \geq u_{p_m}(s_{R_m}^{},s^*_{N-R_m})$$
holds for each given $m \in M,$ any $p_m \in P_m$ and $s_{R_m} \in S_{R_m}.$
\label{def_general}
\end{definition}

\begin{remark}
 If $P=\{\{i\}:i\in N\}$ and $S=\{\{i\}:i\in N\}$ it is clear that the Nash equilibrium is a Berge equilibrium for $P$ relative to $S$, so every Nash equilibrium is a Berge equilibrium. A Berge equilibrium which is also a Nash equilibrium is called Nash-Berge equilibrium \cite{abalo2005}. 
\end{remark}

 If we consider each class $P_i$ of the partition $P$ consists from a player $i$ and each set of  $R_i$ is the set $N$ of players except $i$, we have  $M=N$, $P_i=\{i\}$ and $R_i=N-{i}, \forall i \in N,$ we obtain the Berge-Zhukovskii equilibrium.

Playing in Berge-Zhukovskii sense means that each player wants to maximize the payoff of the other players.  This equilibrium concept can be interpreted as capturing cooperation in a non-cooperative game. Formally we write:

\begin{definition}[Berge-Zhukovskii]
A strategy profile $s^{*} \in S$ is a Berge-Zhu\-kov\-skii equilibrium if the inequality
$$u_i(s^*) \geq u_i( s^*_{i},s_{-i}^{})$$
holds for each player $i=1,...,n,$ and all $s_{-i} \in S_{-i}.$
\end{definition}

The strategy $s^*$ is a Berge-Zhu\-kov\-skii equilibrium, if  the payoff of each player $i$ does not decrease considering any deviation of the other $N-\{i\}$ players.

\begin{example}
Let us consider the Prisoner's Dilemma game presented in Table \ref{table:pd}.

\begin{table}
\caption{The payoff functions of the two players in Prisoner's Dilemma}
\begin{center}
  \begin{tabular}{ l | c | c | c |}
  \multicolumn{4}{c}{\hspace{2.5cm}Player 2} \\ \cline{2-4}
  & & Cooperate  & Defect  \\ \cline{2-4}

{Player 1}  & Cooperate & (2, 2)  & (0, 3)  \\  \cline{2-4}
          &  Defect & (3, 0) & (1, 1) \\ \cline{2-4}

  \end{tabular}
\end{center}
\label{table:pd}
\end{table}

The game has one pure Nash equilibrium  \emph{(Defect, Defect)}, which does not ensure the highest possible payoff for the two players.
In contrary to this, the Berge-Zhukovskii equilibrium of the game is \emph{(Cooperate, Cooperate)}, which may be a better solution for both players.
\end{example}

Inspired by the notion of $\epsilon$-Nash equilibrium, the $\epsilon$-Berge-Zhukovskii equilibrium is introduced. The new equilibrium concept gives a flexibility to the standard Berge-Zhukovskii equilibrium.

The formal definition is the following:

\begin{definition}[$\epsilon$-Berge-Zhukovskii equilibrium]
A strategy profile $s^{*} \in S$ is an $\epsilon$- Berge-Zhu\-kov\-skii equilibrium if the inequality
$$u_i(s^*) \geq u_i( s^*_{i},s_{-i}^{})-\epsilon, \epsilon>0$$
holds for each player $i=1,...,n,$ and $s_{N-i} \in S_{-i}.$
\end{definition}

\begin{remark}
 If $\epsilon=0$ the $\epsilon$-BZ is actually the Berge-Zhukovskii equilibrium.
\end{remark}

We denote by $BZ_{\epsilon}$ the set of all $\epsilon$- Berge-Zhu\-kov\-skii equilibria of the game.

%==============================================================================

\section{Generative relation for $\epsilon$-Berge-Zhukovskii equilibrium}

Generative relations \cite{Lung2008} are used to characterize a certain equilibrium.  Furthermore, they may be used within optimization heuristics for fitness assignment purposes in order to guide their search to the desired equilibrium type.

The first generative relation was introduced for the Nash equilibrium detection \cite{Lung2008}. A generative relation for the detection of Berge-Zhukovskii equilibrium is introduced in \cite{gasko_berge}.

%For the $\epsilon$-Berge-Zhukovskii equilibrium

Consider two strategy profiles $s$ and $q$ from $S$. Denote by $b_{\epsilon}(s,q)$ the number of players who lose (with a deviation of $\epsilon$) by remaining to the initial strategy $s$, while the other players are switching their strategies to $q$ if they are all different from $s$. 

We may express $b_{\epsilon}(s,q)$ as:

%$$b(s,s^*)=card[i \in N, u_i( s^*_{i},s_J)\geq u_i(s^*),\forall J \subseteq N-\{i\},\forall  i=1,...,N],$$
%$$b(s,s^*)=card[i \in N, u_i(s^*) \geq u_i( s^*_{i},s_{N-i})],$$

$$b_{\epsilon}(s,q)=card\{i \in N, u_i(s) < u_i( s_{i}^{},q_{-i})+\epsilon, s_{-i} \textrm{\textdoublebarpipe} q_{-i}\},$$

where $card\{M\}$ denotes the cardinality of the set $M$ and $s_{-i} \textrm{\textdoublebarpipe} q_{-i} \Longleftrightarrow s_j\neq q_j$ for all $j=1,...,n, j\neq  i$.

The intuition behind the construction of $b_\epsilon$ is that, in the search form $\epsilon$-BZ we try to minimize the number of players whose payoff would increase when all the others switch to other strategies. Two strategy profiles may be compared by using the following relation:

\begin{definition}
Let $s,q \in S$.
We say the strategy $s$ is better than strategy $q$ ($s$ dominates $q$) with respect to $\epsilon$-Berge-Zhu\-kov\-skii equilibrium, and we write $s \prec_{B_{\epsilon}}q $, if and only if the inequality  $$b_{\epsilon}(s,q)<b_{\epsilon}(q,s)$$
holds, i.e. there are less players that would benefit when all the others change their strategies from $s$ to $q$ than from $q$ to $s$.
\end{definition}

\begin{remark}
If $b_{\epsilon}(s,q)=b_{\epsilon}(q,s)$ then we consider $s$ and $q$ to be \textit{indifferent} to each other with respect to the $\prec_{B_{\epsilon}}$ relation .
\end{remark}

\begin{definition}
The strategy profile $s^{*} \in S$ is an $\epsilon$-Berge-Zhu\-kov\-skii non-do\-mi\-na\-ted strategy ($BZN_{\epsilon}$), if and only if there is no strategy $s \in S, s \neq s^{*}$ such that $s$ dominates $s^{*}$ with respect to $\prec_{B_{\epsilon}}$ i.e.
 $$s\prec_{B_{\epsilon}} s^{*}.$$
\end{definition}

We may consider relation $\prec_{B_{\epsilon}}$ as a candidate for generative relation of the $\epsilon$-Berge-Zhu\-kov\-skii equilibrium. What we need to prove is that the set of the non-dominated strategies with respect to the relation $\prec_{B_{\epsilon}}$ equals  the set of $\epsilon$-Berge-Zhu\-kov\-skii equilibria of the game. In this case, $\prec_{B_{\epsilon}}$ could be used to compare strategy profiles and guide the search of heuristics such as evolutionary algorithms towards the $\epsilon$-BZ equilibria.

\begin{proposition}
\label{bz0}
If a strategy profile $s^{*} \in S$ is an $\epsilon$- Berge-Zhu\-kov\-skii equilibrium  then the inequality $$b_{\epsilon}(s^{*},s)=0$$
holds, for all $s \in S$.
\end{proposition}

\begin{proof}

Let $s^{*} \in BZ_{\epsilon}$. Suppose there exists a strategy profile $s \in S$, such that $b_{\epsilon}(s^{*},s)=w$, $w >0$.
Therefore there exists $i \in N$, such that $$u_{i}(s^*,s_{-i}^{})+{\epsilon}> u_i(s^*).$$
This contradicts the definition of the $\epsilon$- Berge-Zhu\-kov\-skii equilibrium. Hence $b_{\epsilon}(s^*,s)=0.$
\end{proof}

\begin{proposition}
\label{eq:bz1}
All $\epsilon$-BE equilibrium strategies are $\epsilon$-Berge-Zhu\-kov\-skii non-do\-mi\-na\-ted strategies, i.e.
$$BZ_{\epsilon}\subseteq BZN_{\epsilon}.$$
\end{proposition}

\begin{proof}
Let $s^{*} \in BZ_{\epsilon}$. Suppose $s^{*}$ is dominated. Therefore there exists a strategy profile $s \in S$ dominating $s^*$:
$$s \prec_{B_{\epsilon}} s^*.$$
From definition of the relation $\prec_{B_{\epsilon}}$ we have $$b_{\epsilon}(s,s^{*})<b_{\epsilon}(s^{*},s).$$

As $s^*$ is an $\epsilon$-Berge-Zhu\-kov\-skii equilibrium from Prop. \ref{bz0} it follows that
$$b_{\epsilon}(s^{*},s)=0.$$
Thus we have
$$b_{\epsilon}(s,s^{*})<0.$$
But this is not possible, because $b_{\epsilon}(s,s^{*})$ denotes the cardinality of a set. Therefore $s^*$ is from $BZN_{\epsilon}$ (i.e. non-dominated).
\end{proof}

\begin{proposition}
\label{eq:bz2}
All  $\epsilon$-Berge-Zhu\-kov\-skii non-do\-mi\-na\-ted strategies are $\epsilon$-BZ equilibrium strategies, i.e.
$$BZN_{\epsilon}\subseteq BZ_{\epsilon}.$$
\end{proposition}

\begin{proof}
Let us consider $s^* \in BZN$ ($s^*$ is a non-dominated strategy profile) and suppose that $s^* \not\in BZ.$

If $s^* \not\in BZ$ $\Rightarrow \exists s_{-i}$ such that 
\begin{equation}
 u_i(s^*)< u_i(s^*,s_{-i})+\epsilon, \label{nedom}
\end{equation}
and  $s^{*}_{-i} \textrm{\textdoublebarpipe} s_{-i}.$

Let us denote by $q$ the strategy profile $(s^*,s_{-i}).$

We have: $$b_{\epsilon}(s^*,q)=card\{j \in N, u_j(s^*)<u_j(s^*_{j}, q_{-j})+\epsilon, s^*_{-j} \textrm{\textdoublebarpipe} q_{-j} \}.$$

But for all $ j \neq i$ we have $s^*_i=q_i$ so $s^*_{-j} \textrm{\textdoublebarpipe} q_{-j}$ does not hold, and for $j=i$ relation (\ref{nedom}) holds,  therefore $b_{\epsilon}(s^*,q)=1.$

On the other hand $$b_{\epsilon}(q,s^*)=card\{j \in N, u_j(q)<u_j(q_j, s^{*}_{-j}) +\epsilon, q_{-j} \textrm{\textdoublebarpipe} s^*_{-j} \}.$$

If $i \neq j$ $s^*_{i}=q_{i}$ so $q_{-j} \textrm{\textdoublebarpipe} s^*_{-j}$ does not hold.

If $i =j$ we have $(q_i, s^*_{-i})=s^*,$ therefore $u_i(q)<u_i(q_i, s^*_{-i}),$ if and only if $$u_i(q)<u_i(s^*)+\epsilon$$ if $$ u_i(s^*,s_{-i})<u_i(s^*)+\epsilon$$ which would contradict relation (\ref{nedom}). Therefore  $b_{\epsilon}(q,s^*)=0.$

We have that $b_{\epsilon}(s^*,q)>b_{\epsilon}(q,s^*)$, which contradicts the assumption of non-domination, therefore $s^*$ is an $\epsilon$-BZ equilibrium.
\end{proof}

\begin{proposition}
All $\epsilon$-Berge-Zhu\-kov\-skii equilibria are $\epsilon$-Berge-Zhu\-kov\-skii non-dominated strategies and all $\epsilon$-Berge-Zhu\-kov\-skii non-dominated strategies are $\epsilon$-Berge-Zhu\-kov\-skii equilibria: $$ BZ_{\epsilon}=BNS_{\epsilon}.$$
\end{proposition}

\begin{proof}
Directly from Proposition \ref{eq:bz1} and Proposition \ref{eq:bz2}.
\end{proof}

\section{Evolutionary detection method}

Solving a multiplayer game in which players seek to maximize their payoffs has many common features with the Multiobjective Optimization Problem (MOP). In both cases the goal is to maximize the payoffs/objective functions.

For a MOP the \textit{m} objectives are represented by a set $\{u_i\}_{i\in \{1, ..., m\}}$ of functions where $u_i : S \rightarrow\mathbb{R}$ maps a solution \textit{s} from the decision space $S\subseteq\mathbb{R}^n$ to the objective space $\mathbb{R}$. $F:S\rightarrow\mathbb{R}^m$ represents the objective vector that needs to be maximized $F(s)=(u_1(s), u_2(s), ..., u_m(s))$. 

Usually the ideal solution vector $s$ that optimizes simultaneously all objective functions does not exist. When solving a multiobjective problem usually one does not find a single solution that best approximate $F$ but a set of solutions that approximate the \textit{Pareto optimal set}. The Pareto optimal set is formed by the \textit{Pareto non-dominated} solutions which represent the best trade-offs among the \textit{m} objectives.
Thus two solutions are compared using the Pareto dominance relation: for any two decision vectors $s,s'\in S$ we say that $s$ is better than $s'$ or $s$ \textit{dominates} $s'$, if $u_i(s)\geq u_i(s') \forall i\in\{1, ..., m\}$ and $\exists j\in\{1, ..., m\}$ such that $u_j(s)>u_j(s')$.

Evolutionary Multiobjective Algorithms have been successfully used for solving such problems \cite{Coello2007}, \cite{Zhou2011}. As they are population based metaheuristics they represent a good choice because by evolving a set of possible solutions (the population) a good approximation of the Pareto front can be found in a single run while the shape, continuity or other mathematical properties of the true front do not hinder the search. 

The task of the optimization algorithm is to find a good approximation of the Pareto optimal set while maintaining a good diversity in the population. In Pareto based multiobjective optimization algorithms the search is driven by the Pareto dominance relation: if a new generated solution Pareto dominates a solution from the current population it replaces it in the next population. 

For detecting the $\epsilon$-Berge-Zhukovskii equilibrium any Pareto based multiobjective algorithm is suitable. The only modification needed is the replacement of the Pareto dominance relation whenever it is used during the search %(if the algorithm uses Pareto ranking for maintaining diversity)
with the relation $\prec_{B_{\epsilon}}$.

%==============================================================================
\section{Numerical experiments}
The computation of $\epsilon$-Berge-Zhukovskii equilibria is illustrated for two examples of games with two and three players. These games correspond to a multiple objective optimization problem where the problem objectives are represented by the payoff of each player but searching for a different \textit{solution concept}.

For detecting $\epsilon$ -Berge-Zhukovskii equilibria a modified version of the Nondominated Sorting Genetic Algorithm II (\textit{NSGA-II}) \cite{Deb2002} algorithm is considered. \textit{NSGA-II}  is an evolutionary multiobjective optimization algorithm %with the variation operator $v$ based on Differential Evolution \cite{Storn97} and 
based on the Pareto dominance. The \textit{NSGA-II} algorithm is modified by replacing the Pareto dominance procedure with the generative relation $\prec_{B_{\epsilon}}$. A new version of the \textit{NSGA-II} method  for $\epsilon$-Berge-Zhukovskii equilibrium detection called \textit{BZ-NSGA-II} is obtained.

For all the tests a population of 150 individuals is used for 150 generations of the \textit{BZ-NSGA-II}. As for the variation operators we use a distribution indexes for mutation and crossover $\eta _m=20$ and $\eta _c=20$.
%the scheme \textit{DE/rand/1/bin} proposed in \cite{Storn97} is used with a mutation scaling factor $F=0.2$ and a crossover probability $CR=0.7$.

The obtained results are illustrated in the following manner: for two or three players the payoffs space is represented by assigning axes to the payoff of each player and representing solutions as points in the two and three dimensional spaces respectively. In each graphic the set of $\epsilon$-BZ is represented with gray color.

\subsection{Experiment 1}

Let us consider the two-person continuous game $G_1$ \cite{nessah_tian_working}, having the following payoff functions:
\[
 \left.\begin{aligned}
        &u_1(s_1,s_2) = -s_{1}^2-s_{1}+s_2,\\
        &u_2(s_1,s_2) = 2s_{1}^2+3s_1-s_{2}^2-3s_2,\\
        &s_i \in [-2,1], i=1,2.
       \end{aligned}
 \right.
\]

The Berge-Zhukovskii equilibrium of the game is $(1,1)$ with the corresponding payoffs $(-1,1).$ Figures \ref{g1_e01}-\ref{g1_e09} present the detected $\epsilon$-Berge-Zhukovskii equilibria for $\epsilon \in \{0, 0.1, 0.2, 0.5, 0.9\}$. 
%\begin{figure}
%\begin{minipage}[b]{0.48\linewidth}
%		\includegraphics[scale=0.4]{g1_new_1}
%    \caption{Payoffs for game $G_1$, $\epsilon =0.1$.}\label{fig:Game1detail1}
%\end{minipage}
%\begin{minipage}[b]{0.48\linewidth} 
%    \includegraphics[scale=0.4]{g1_new_2}
%    \caption{Payoffs for game $G_1$, $\epsilon = 0.2$.}\label{fig:Game1detail2}
%\end{minipage}
%\end{figure}

%\begin{figure}
%\begin{minipage}[b]{0.48\linewidth}
%		\includegraphics[scale=0.4]{g1_new_3}
%    \caption{Payoffs for game $G_1$, $\epsilon=0.5$.}\label{fig:Game1detail3}
%\end{minipage}
%\begin{minipage}[b]{0.48\linewidth} 
%    \includegraphics[scale=0.4]{g1_new_4}
%    \caption{Payoffs for game $G_1$, $\epsilon=0.9$.}\label{fig:Game1detail4}
%\end{minipage}
%\end{figure}

\begin{figure}[t]
\begin{minipage}[b]{0.48\linewidth}
	\includegraphics[scale=0.4]{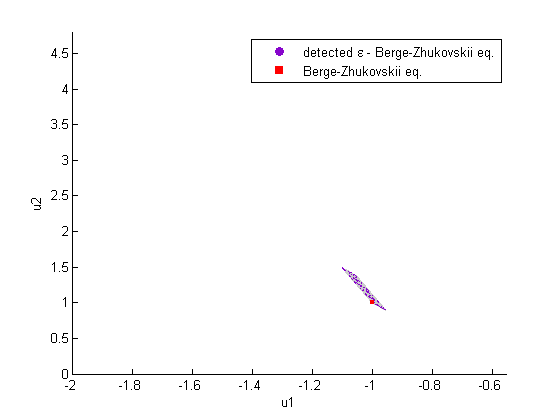}
   	\caption{$G_1$, $\epsilon=0.1$.  Theoretical $\epsilon$-BZ are represented in gray. Detected solutions cover the theoretical front efficiently }\label{g1_e01}
\end{minipage}
\hspace{0.04\linewidth}
\begin{minipage}[b]{0.48\linewidth}
	\includegraphics[scale=0.4]{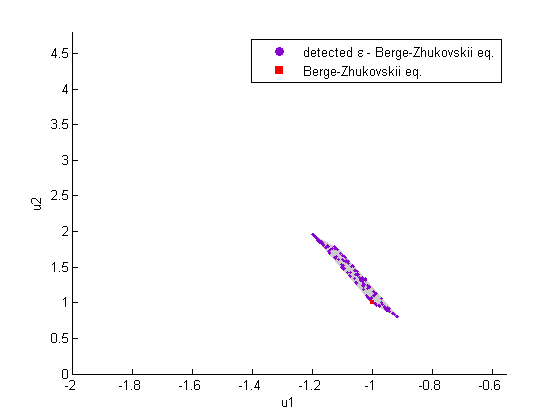}
    \caption{$G_1$, $\epsilon=0.2$.  Increasing the value of $\epsilon$ leads to a larger set of $\epsilon$-BZ which is covered by the detected solutions.}\label{g1_e02}
\end{minipage}
\end{figure}
\begin{figure}[t]
\begin{minipage}[b]{0.48\linewidth}
	\includegraphics[scale=0.4]{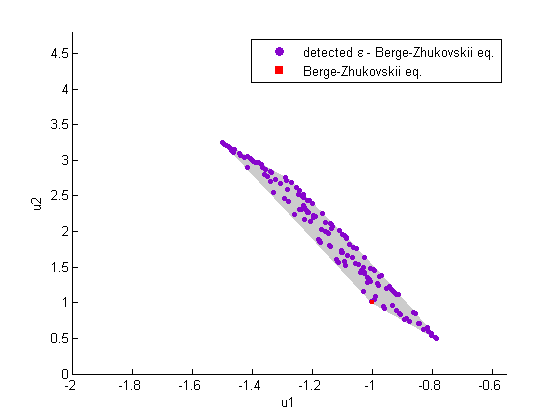}
   	\caption{$G_1$, $\epsilon=0.5$. Increasing $\epsilon$ enlarges the set of equilibria, but it is still well covered by detected solutions.}\label{g1_e05}
\end{minipage}
\hspace{0.04\linewidth}
\begin{minipage}[b]{0.48\linewidth}
	\includegraphics[scale=0.4]{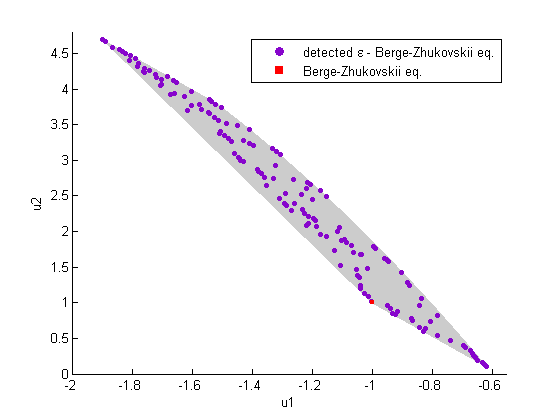}
    \caption{$G_1$, $\epsilon=0.9$. Even for $\epsilon=0.9$ $BZ-NSGA-II$ is capable to compute a reasonable set of $\epsilon$-BZ solutions. }\label{g1_e09}
\end{minipage}
\end{figure}

\subsection{Experiment 2 - Voluntary contribution mechanism}

The Voluntary Contribution Mechanism (\textit{VCM}) is a good example illustrating that people are not totally self-interested. They spend time making something for the common good. %Several experiments are made concerning on the study of player behavior \cite{Fischbacher2001}, \cite{Muller2008}. For a theoretical studies see for example \cite{bochet}. A model of the VCM is described as game $G_2$:
Theoretical studies \cite{Bochet2006} and experiments \cite{Fischbacher2001,Muller2008} are made concerning player behavior. A model of the \textit{VCM} is described as game $G_2$:
\[u_i(s)=10-s_i+0.4\sum_{i=1,n}s_i, s_i \in [0,10], i=1,...,n.\]

In this game the Berge-Zhukovskii equilibrium is achieved when all players play strategy $10$, which means they spend all for the public good. Detected $\epsilon$-Berge-Zhukovskii equilibria for the two-player version of the \textit{VCM} game are presented in Figures \ref{g2_e01d}, \ref{g2_e02d}, \ref{g2_e05d} and \ref{g2_e09d}. Results obtained for the three player version of the game are depicted in Figures \ref{g3_e01}, \ref{g3_e02}, \ref{g3_e05}  and \ref{g3_e09}.

%Figures \ref{fig:VCgame2JucatoriTeoretic} and \ref{fig:VCgame2JucatoriEvolutionary} present a comparison between exact and evolutionary obtained $\epsilon$-Berge-Zhukovskii solutions for the VCM game ($\epsilon=0.1$). It can be observed, that the evolutionary method gives a good approximation of the theoretical solutions. 

\begin{figure}[t]
\begin{minipage}[b]{0.48\linewidth}
	\includegraphics[scale=0.4]{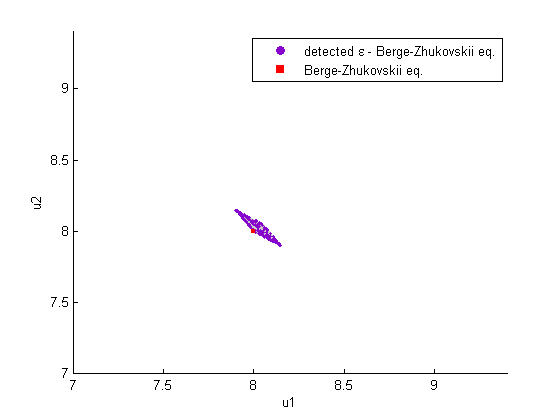}
   	\caption{Voluntary Contribution Mechanism, evolutionary detected solutions for $\epsilon=0.1$.}\label{g2_e01d}
\end{minipage}
\hspace{0.04\linewidth}
\begin{minipage}[b]{0.48\linewidth}
	\includegraphics[scale=0.4]{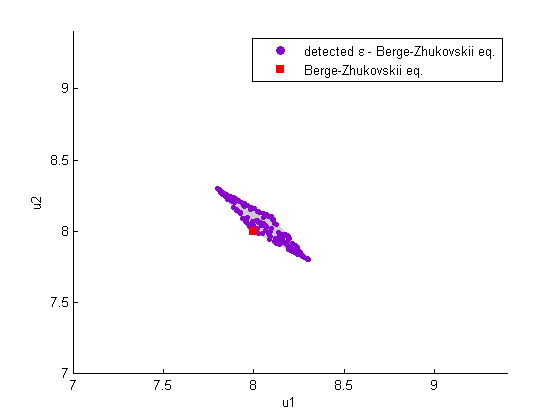}
    \caption{Voluntary Contribution Mechanism, evolutionary detected solutions for $\epsilon=0.2$.}\label{g2_e02d}
\end{minipage}
\end{figure}
\begin{figure}[t]
\begin{minipage}[b]{0.48\linewidth}
	\includegraphics[scale=0.4]{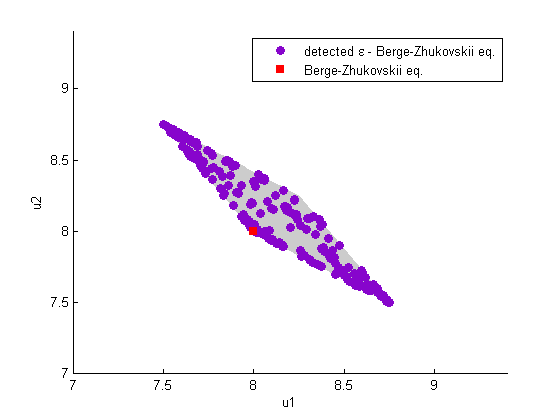}
   	\caption{Voluntary Contribution Mechanism, evolutionary detected solutions for $\epsilon=0.5$.}\label{g2_e05d}
\end{minipage}
\hspace{0.04\linewidth}
\begin{minipage}[b]{0.48\linewidth}
	\includegraphics[scale=0.4]{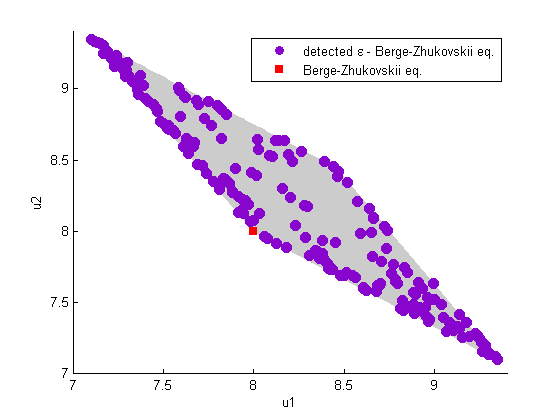}
    \caption{Voluntary Contribution Mechanism, evolutionary detected solutions for $\epsilon=0.9$.}\label{g2_e09d}
\end{minipage}
\end{figure}

%\begin{figure}
%\begin{center}
%	\includegraphics[scale=0.6]{vc3jucatori}
%   	\caption{Payoffs for Voluntary Contribution Mechanism for 3 players, $\epsilon \in \{0, 0.1, 0.2, 0.5, 0.9\}$.}\label{fig:VCgame3Jucatori}
%\end{center}
%\end{figure}

%\begin{figure}
%\begin{minipage}[b]{0.48\linewidth}
 % 	\includegraphics[scale=0.4]{vc3jucatoriDetail}
 %   \caption{Payoffs for game Voluntary Contribution Mechanism for 3 players -detail-, $\epsilon \in \{0.1, 0.2, 0.5\}$.}\label{fig:VCgame3JucatoriDetail}
%\end{minipage}
%\hspace{0.04\linewidth}
%\begin{minipage}[b]{0.48\linewidth}
%    \includegraphics[scale=0.4]{vc3jucatoriDetailVar1}
%    \caption{Payoffs sfor game Voluntary Contribution Mechanism for 3 players -detail-, $\epsilon \in \{0.1, 0.2\}$.}\label{fig:VCgame3JucatoriDetailVar2}
%\end{minipage}
%\end{figure}

\begin{figure}[t]
\begin{minipage}[b]{0.48\linewidth}
	\includegraphics[scale=0.4]{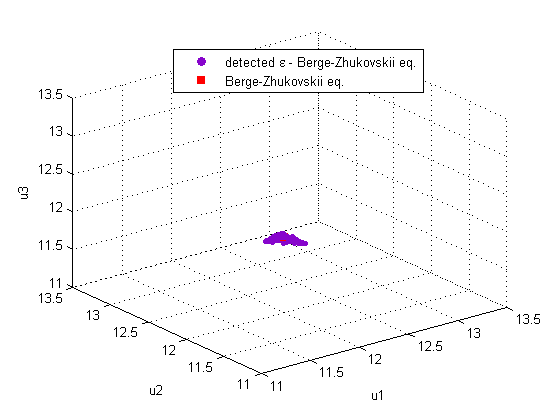}
   	\caption{Voluntary Contribution Mechanism and evolutionary detected solutions for $\epsilon=0.1$.}\label{g3_e01}
\end{minipage}
\hspace{0.04\linewidth}
\begin{minipage}[b]{0.48\linewidth}
	\includegraphics[scale=0.4]{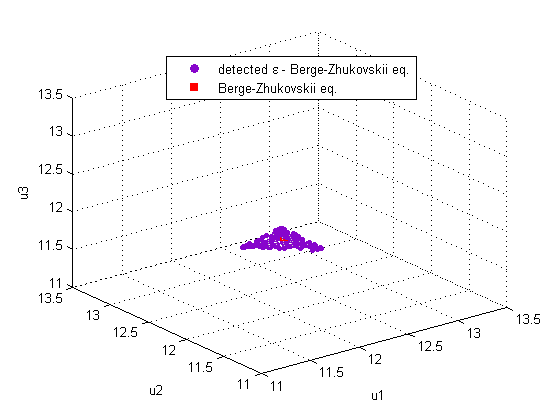}
    \caption{Voluntary Contribution Mechanism and evolutionary detected solutions for $\epsilon=0.2$.}\label{g3_e02}
\end{minipage}
\end{figure}

\begin{figure}[t]
\begin{minipage}[b]{0.48\linewidth}
	\includegraphics[scale=0.4]{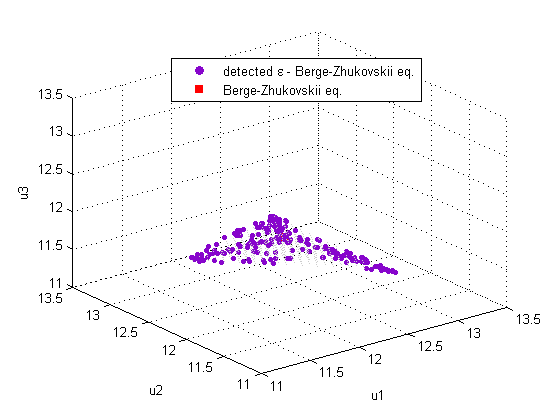}
   	\caption{Voluntary Contribution Mechanism and evolutionary detected solutions for $\epsilon=0.5$.}\label{g3_e05}
\end{minipage}
\hspace{0.04\linewidth}
\begin{minipage}[b]{0.48\linewidth}
	\includegraphics[scale=0.4]{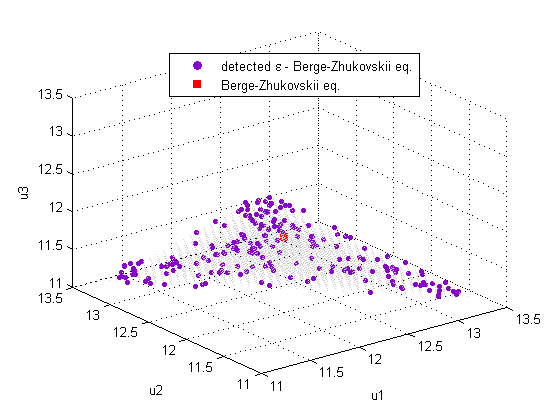}
    \caption{Voluntary Contribution Mechanism and evolutionary detected solutions for $\epsilon=0.9$.}\label{g3_e09}
\end{minipage}
\end{figure}

\paragraph{Discussion}
Numerical experiments presented here illustrates that \textit{BZ-NSGA-II} technique is able to find a good approximation of the $\epsilon$-Berge-Zhukovskii set for different values of $\epsilon$ and for different number of players. 

As it is natural to expect, by increasing the value of $\epsilon$ the number of $\epsilon$-Berge-Zhukovskii equilibria also increases by covering a region (whose shape depends on the payoff function) in the payoffs space that includes the Berge-Zhukovskii equilibrium. In fact, for  $\epsilon=0$ the Berge-Zhukovskii equilibrium is obtained.

%==============================================================================

\section{Conclusion}

Berge-Zhukovskii equilibrium is a powerful concept specially in trust games. The $\epsilon$-Berge-Zhukovskii equilibrium represents a flexible concept that approximates the BZ equilibrium. It may be also considered as a relaxation of the BZ equilibrium. The $\epsilon$-Berge-Zhukovskii equilibrium is introduced in this paper. A generative relation characterizing the set of $\epsilon$-Berge-Zhukovskii equilibria is also proposed. 

Apart defining the equilibrium concept, a computational method to approach these equilibria is presented. Based on the idea that the equilibrium search and the multi-objective optimization can be considered in the same class of problems, an evolutionary algorithm for multiobjective optimization is adapted for $\epsilon$-Berge-Zhukovskii equilibria detection. Numerical experiments validate the proposed method and confirm the theoretical results presented.

%==============================================================================
\section{Acknowledgments}

This project was supported by the national project code TE 252 financed by the Romanian Ministry of Education and Research CNCSIS-UEFISCSU. The first author wishes to thank to the "Collegium Talentum". The authors would also like to acknowledge the support received within the OPEN-RES Academic Writing project 212/2012.

%==============================================================================


\bibliographystyle{elsarticle-num}

\end{document}